\documentclass[10pt]{article}

\usepackage{geometry}
\usepackage{enumerate}
\usepackage{amsmath}
\usepackage{amsthm}
\usepackage{hyperref}
\usepackage{graphicx}

\newtheorem{theorem}{Theorem}

\newtheorem{conjecture}[theorem]{Conjecture}
\newtheorem{corollary}[theorem]{Corollary}

\theoremstyle{definition}
\newtheorem{definition}[theorem]{Definition}

\def\({\left(}
\def\){\right)}

\title{Analysis of a Poisson-picking symmetric winners-take-all game with randomized payoffs}
\author{Abel Molina}

\begin{document}

\maketitle

\begin{abstract}
Winners-take-all situations introduce an incentive for agents to diversify their behavior, since doing so will result in splitting an eventual price with fewer people. At the same time, when the payoff of a process depends on a parameter choice that is symmetric with respect to agents, all agents have the incentive to choose the values of the parameter that lead to higher payoffs. We explore the trade-off between these two considerations, with a focus on a particular example. This example can be seen as a simple model for the situation where a group of friends bet against each other about the top-scoring team in a sports league. We obtain analytic characterizations of the symmetric equilibria in the case of only $2$ agents and in the case where there are only two possible top-scorers. We also conduct some simulations beyond these cases, and observe how does the pressure to diversify behavior evolve as the parameters of the model change.
\end{abstract}

\section{Introduction}
In \textbf{winners-take-all} situations, there is pressure for agents to diversify their behavior. An extreme example of this diversification pressure can be seen in LUPI lottery games \cite{ostling2011testing}, where the payoffs are determined by a method along the following lines:

\begin{quote}
Take one money unit from everyone who chooses an integer. The person who chooses the smallest integer such that no one else chose the same integer gets 80\% of those money units. Everyone else does not receive any money.
\end{quote}

\noindent The fact that the winner corresponds to the unique smallest chosen number means that there is pressure towards picking smaller numbers. However, the single winner constraint gives us an incentive not to select numbers that we think other people might select, which trades off with the previous preference for smaller numbers.

A more nuanced simlar situation where several winners are possible, but there is still pressure to diversify behavior, is the following one:

\begin{quote}
Three friends bet on whether team A or team B will win a sports game, with ties in the game not possible. To do so, each of them puts a money unit into a common pool, and makes a prediction of whether team A or team B will win. Then, if anyone chose the right winner, they will split those money units (we asume the units are arbitrarily divisible). Otherwise, everyone will get back their money unit. Team A is the favorite, with a 60\% winning chance.
\end{quote}

\noindent One can verify that everyone always picking the favorite as the winner is not a Nash equilibrium. This is because if two agents always pick the favorite team, the expected payoff for the third agent when picking the underdog is $0.2$, but is just $0$ when picking the favorite team. This contradicts a ``greedy" intuition, which would reason that we only have a positive payoff when we select the team who wins, and therefore we should aim to select the winning team as often as possible.

This example falls within the camp of  \textbf{parimutuel} betting - this is the term for situations where all bets about a particular outcome are placed together in a pool, any taxes and house fees are then removed (in the examples we explore there is no such thing), and finally the payoffs are calculated by sharing the pool among all winning bets. 

Furthermore, this example extends naturally to the example we focus our study on, which is a tweak of the previous situation. In the example we focus on, those friends are not betting on the outcome of a competition, and instead they are trying to guess which team will be the top goal-scorer (or point-scorer) in a league. We model this through the following process:

\begin{quote}
A group of friends each bet on which of several independent Poisson processes will result in a larger number of events when running for a time unit.  To do so, each of them puts a money unit into a common pool.  Then, the money in the pool is divided uniformly among the friends who picked the processes with the most events, out of those processes that have been selected by at least one person.
\end{quote}

This assumes that Poisson processes are an appropriate choice to model goal-scoring in the game at hand. The reasonable accuracy of Poisson distributions as a simple model for the total number of goals in a sports game is discussed in \cite{maher1982modelling, dixon1997modelling} for the case of association football, and  \cite{mullet1977simeon, thomas2007inter}  for the case of ice hockey. There is no published research to our knowledge that studies the fit of individual scoring tallies to Poisson distributions, but it does seem like an intuitive choice for a simple model as well. However, when looking into individual scoring tallies one could reasonably expect to see correlations between players in the same team, which contradicts the assumption of independence. Our model is likely then not to adjust accurately to those situations. Note that such correlations could also reasonably exist between the goals scored by teams facing each other, but over the period of a whole competition usually only a small fraction of a team's games will be against any particular opponent, so this is not a strong issue in terms of our initial interpretation.

Parimutuel betting has been considered previously in the academic literature, with some of the best-known studies looking into aspects such as:

\begin{enumerate}[(a)]

\item Biases exhibited by participants in several parimutuel betting markets involving horse-racing \cite{thaler1988anomalies}.

\item How the bias towards underdogs on the side of some bettors might be partially caused by the house taking a cut of the pool \cite{hurley1995note}.

\item Interpreting parimutuel betting settings as prediction markets \cite{plott2003parimutuel, pennock2004dynamic}.

\end{enumerate}

One can also look at our problem in the wider context of situations where there is pressure for an agent to diversity its actions away from a greedy choice. A well-known category of these are multi-armed bandit situations \cite{robbins1985some}, where the uncertainty about an agent's own model leads to the pressure for diversification. Another similar situation is the one in \cite{chawla2014mechanism}, where an auctioneer has the incentive to conduct non-optimal auctions in order to gain knowledge about bidders, and therefore increase the revenue of future auctions.

The setting here is also slightly reminiscent of the setting of all-pay auctions, since we also have all agents making a payment, but only some of them receiving some positive utility in return. The spirit of our results is then similar in part to that of the results in \cite{dulleck2006all}, which show how budget-constrained agents are sometimes better off engaging in risky behavior in order to improve their outcome in the context of  all-pay auctions.

We proceed now to the study of our game. In Section \ref{sec:rigurous}, we will give a definition of it that is mathematically precise, and places it within the context of a larger class of games. We will generally focus on our Poisson-picking example, but we will point it out when our results easily extend to the larger class. Then, in Section \ref{sec:anstudy}  we will approach from an analytical point of view the cases where the number of agents or Poisson processes is equal to $2$, and additionally consider a setting where other agents are known to behave greedily. Finally, in in Section \ref{sec:practicalAnalysis} we will consider some larger cases through simulation procedures.

We will analyze the problem from the point of view where agents are trying to maximize their expected payoff. Therefore, we will concern ourselves mostly with Nash equilibria, and we will particularly focus on symmetric mixed Nash equilibria in our theoretical study of equilibrium situations. This is because the payoff matrix of the situations we consider is symmetric, and therefore the foundational results of Nash \cite{nash1951non} establish that there will always be a symmetric mixed Nash equilibrium. Note that the situations we consider are symmetric zero-sum games, so the expected value of a Nash equilibrium is zero for everyone. Also, note that any symmetric choice of strategies has those optimal expected values ($0$ for everyone). However, as one can see in our initial example with the three betting friends, not all of those symmetric choices are Nash equilibria.  Note also that by linearity of expectation, when looking for Nash equilibria we  can replace the random variable corresponding to the payoffs (given a fixed choice of processes by the agents) by its expectation.

\section{Rigorous model definition}

\label{sec:rigurous}

\begin{definition}
\label{def:general}
Given integers $n \geq 2, m \geq 2$, and a random variable $X$ taking values in the set of ordered partitions of $\{1, \ldots, m\}$, the \textbf{winners-take-all pool} $W$ with parameters $(n, m, X)$ is the following game:

\begin{enumerate}

\item Each agent, labeled by a distinct integer $i \in \{1, \ldots, n\}$, donates 1 money unit to a common pool, and makes a choice of of an integer $c_i \in \{1, \ldots , m\}$. We denote by $C$ the set of all integers that have been chosen by at least one agent.

\item We sample a value from $X$. This value will be equal to a partition of $\{1, \ldots, m\}$, which we can write as $R  = (R_1, \ldots,  R_k)$.

\item Let $R_w$ be such that $w = \min(\{ l :  R_l \cap C \neq \emptyset \})$. Those agents who selected a value $c_i \in R_w$ split the pool between them. The total payoff for them is then $ \dfrac{n}{| i: c_i \in R_w |} - 1 $, while for everyone else the payoff is $-1$.

\end{enumerate}
\end{definition}

\begin{definition}
A \textbf{Poisson-picking pool} $PP$ is specified by parameters $(n, m, Y_1(\lambda_1),  \ldots, Y_m(\lambda_m))$, where $Y_1, \ldots, Y_m$ are $m$ random variables corresponding to $m$ independent Poisson processes with rates $\lambda_1, \ldots, \lambda_m$, respectively. Then,  $PP$  will be a winners-take-all pool  with parameters $(n, m, X(Y_1, \ldots, Y_m))$ where a value of $X(Y_1, \ldots, Y_m)$ is drawn using the following process:

\begin{enumerate}

\item We draw a value from each of $Y_1, \ldots, Y_m$, obtaining $m$ values $y_1, \dots, y_m$.

\item We return the non-empty level sets $L_y = \{i : y_i = y\} $, ordered  from the highest  drawn value of $y$ to the lowest.

\end{enumerate}
\end{definition}

For an agent with label $i$, we will use $s_i$ to refer to the variable corresponding to its strategy. We will allow these strategies to be mixed strategies, with $s_i(j)$ being the chance that agent $i$ selects value $j$. Similarly, we will use $P_i$ to refer to the payoff of the agent with label $i$.

\section{Analytical study}

\label{sec:anstudy}

\subsection{Adjusting behavior to greedy choices from other agents}

\label{sec:adjusting}

One can consider the situation where an agent has a poor opinion of the reasoning skills of other agents, and therefore is quite confident that all of them are going to choose the Poisson process with the highest rate. We explore this now for the simple case of two Poisson processes, and obtain the following result:

\begin{theorem} \label{thm:f1} Consider a Poisson-picking pool with parameters $(n, 2, Y_1(\lambda_1), Y_2(\lambda_2))$. Assume that agents $\{1, \ldots, n -1\}$ choose process $1$. When prob$(Y_2 > Y_1) > prob(Y_2 < Y_1)/(n-1)$, it is uniquely optimal for agent $n$ to choose process $2$. When prob$(Y_2 > Y_1) < prob(Y_2 < Y_1)/(n-1)$, it is uniquely optimal for agent $n$ to choose process $1$. When prob$(Y_2 > Y_1) = prob(Y_2 < Y_1)/(n-1)$, any mixed strategy is optimal.
\end{theorem}

\begin{proof}
If $n-1$ agents select process $1$ with probability $1$, the expected payoff $P_n$ of agent $n$ will be given by

\[
E(P_n) = s_n(1)*\(\dfrac{n}{n} -1\) + s_n(2)*(-1 + \dfrac{n}{n}*\text{prob}( Y_2 = Y_1) + n*\text{prob}( Y_2 > Y_1) ).
\]

\noindent We can see from this formula that if and only if $\text{prob}( Y_2 > Y_1) >  \text{prob}( Y_2 \neq Y_1)/n$, it will be uniquely optimal for the agent to always choose process $2$. This condition is equivalent to $\text{prob}( Y_2 > Y_1) >  \text{prob}( Y_2 < Y_1)/(n-1)$ The same happens with respect to process $1$, and the reverse condition $\text{prob}( Y_2 > Y_1) < \text{prob}( Y_2 < Y_1)/(n-1)$. In the the border case $\text{prob}( Y_2 > Y_1) = \text{prob}( Y_2 < Y_1)/(n-1)$ we have that any strategy is optimal.

\end{proof}

We aim to give in Figure \ref{fig:skellamPlot1} some intuition of what this means. In that figure, we consider several fixed choices of $n$, and plot the curve on the $(\lambda_1, \lambda_2)$ plane induced by the condition $\text{prob}( Y_2 > Y_1) =  \text{prob}( Y_2 < Y_1)/(n-1)$. This can be done using the equivalent condition $\text{prob}( Y_2 - Y_1 > 0) = \text{prob}( Y_2 - Y_1 < 0)/(n-1)$, and the fact that  $Y_2 - Y_1$ will be given by a Skellam distribution with the rates of the processes as its parameters.  One can observe that for $n=2$, it is always optimal to select the favorite.

\begin{figure}
\begin{center}

\includegraphics[height=10.5cm]{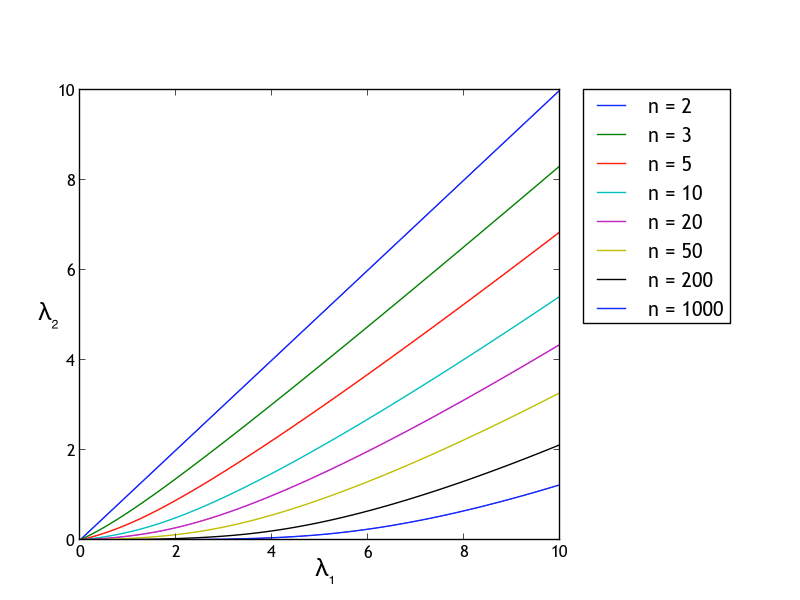}
\end{center}
  \caption{Plots in the ($\lambda_1$, $\lambda_2$) plane for the boundaries between it being optimal to select the favorite or the underdog in the context of Section \ref{sec:adjusting}. These boundaries are given by the condition $\text{prob}( Y_2 > Y_1) =  \text{prob}( Y_2 < Y_1)/(n-1)$.  \label{fig:skellamPlot1} }
\end{figure}

Note that the proof of Theorem \ref{thm:f1} extends to the case of general winners-take-all pools, by looking at what the sign of $Y_2 - Y_1$ represents in terms of the partitions in Definition \ref{def:general}. One obtains then the following theorem:

\begin{theorem}
Consider a winners-take-all pool $(n, 2, X)$. Assume that agents $\{1, \ldots, n -1\}$ choose process $1$. When prob$(X = \{2\}\{1\}) > prob(X =\{1\}\{2\})/(n-1)$, it is uniquely optimal for agent $n$ to choose process $2$. When prob$(X = \{2\}\{1\}) < prob(X =\{1\}\{2\})/(n-1)$, it is uniquely optimal for agent $n$ to choose process $1$.~When prob$(X = \{2\}\{1\}) = prob(X =\{1\}\{2\})/(n-1)$, any mixed strategy is optimal.
\end{theorem}

Note also that the  analysis in Theorem \ref{thm:f1} can be extended easily to the situation where agent $n$ can choose between an arbitrary number of processes. This is because the processes are independent of each other, and we know that the price would only be shared in case of choosing process $1$. Therefore, choosing a process that has the highest rate when ignoring process $1$ will be a weakly dominant strategy with respect to choosing any process other than process $1$. We obtain then the following corollary:

\begin{corollary}
Consider a Poisson-picking pool $(n, m, Y_1(\lambda_1), Y_2(\lambda_2), \ldots, Y_m(\lambda_m) )$. Assume that agents $\{1, \ldots, n -1\}$ choose process $1$. Let $Y_i$ be a highest rate process when ignoring process $1$. When prob$(Y_i > Y_1) \geq prob(Y_i < Y_1)/(n-1)$, it is optimal for agent $n$ to choose process $i$. When prob$(Y_i > Y_1) < \text{prob}(Y_i < Y_1)/(n-1)$, it is uniquely optimal for agent $n$ to choose process $1$.
\end{corollary}

\subsection{Nash equilibria}

\subsubsection{Only two agents}

We look first at symmetric Nash equilibria for the case $(2, m, Y_1(\lambda_1), \ldots, Y_m(\lambda_m))$, where we assume without loss of generality than $\lambda_1 \geq \lambda_2 \geq \ldots \geq \lambda_m$. Our main result is the following:

\begin{theorem}
\label{th:twoagents}
Consider any Poisson-picking game of the form $(2, m, Y_1(\lambda_1), \ldots, Y_m(\lambda_m))$, where we have two agents and $\lambda_1 \geq \lambda_2 \geq \ldots \geq \lambda_m$. Then, the best response for each agent is to select an arbitrary mix of processes with the highest rate, no matter what the other agent does.\end{theorem}

\begin{proof}
The expected payoff for agent $1$ when he chooses process $j$ and agent 2 chooses process $k$ will be $-1 + 0*\text{prob}(Y_j < Y_k) + 1*\text{prob}(Y_j = Y_k) + 2*\text{prob}(Y_j > Y_k)$

Then, if we consider an arbitrary strategy $s_1(1), \ldots, s_1(m)$ for agent $1$, we have that the expected payoff for agent $1$ when agent $2$ chooses process $k$ is

\begin{align}
 & \sum_{j=1}^m s_1(j)  (-1 + 0*\text{prob}(Y_j < Y_k) + 1*\text{prob}(Y_j = Y_k) + 2*\text{prob}(Y_j > Y_k) ) \\
 = & - 1 +  \sum_{j=1}^m  s_1(j)  (\text{prob}(Y_j = Y_k) + 2*\text{prob}(Y_j > Y_k))  \label{eq:secondline} \\
 = & - 1 + s_1(k) +  \sum_{\substack{j=1 \\ j\neq k }}^m  s_1(j)  (\text{prob}(Y_j = Y_k) + 2*\text{prob}(Y_j > Y_k)) \\
 = & - 1 + s_1(k) +  \sum_{\substack{j=1 \\ j\neq k }}^m  s_1(j)  (\text{prob}(Y_j = Y_k) + \text{prob}(Y_j > Y_k) + (1  - \text{prob}(Y_j = Y_k)  \nonumber   \\
  &  - \text{prob}(Y_j < Y_k) )) \\
 = &  \sum_{\substack{j=1 \\ j\neq k }}^m  s_1(j)  ( \text{prob}(Y_j > Y_k)  - \text{prob}(Y_j < Y_k) ))  \\
  = &  \sum_{j=1}^m  s_1(j)  ( \text{prob}(Y_j > Y_k)  - \text{prob}(Y_j < Y_k) )) 
 \label{eq:final}
\end{align}

It is clear that no matter the strategy of agent $1$, the value in \eqref{eq:final} is minimized by agent $2$  choosing $k=1$, or another value of $k$ corresponding to a process with the same rate. By the zero-sum character of the game and the linearity of expected values, it follows that the best responses for agent $2$ are the mixed combinations of these strategies, independently of the strategy of agent $1$. By symmetry, the same analysis applies in the other direction. This proves our claim.

\end{proof}

Note that Theorem \ref{th:twoagents} implies that it is a symmetric Nash equilibrium for all agents to select process $1$ with probability $1$. However, any pair of strategies with support on the processes with $\lambda_j = \lambda_1$ represent a Nash equilibrium, so when there are several process with the highest rate there can be non-symmetric Nash equilibria as well.

Note also that the proof of Theorem \ref{th:twoagents} does not extend beyond Poisson-picking pools to the case of general winners-take-all pools (replacing in that extension the processes with the highest rate by the integers with the highest chance of being in the first set of the partition). This is because in the distribution of ordered partitions from which we sample there could be correlations between the position in the partition of different integers. To see a concrete example, consider the winners-take-all pool with $n=2$ and $m=5$ where the partition is drawn at random from the four options $(\{5\}\{1\}\{2\}\{4\}\{3\},\{5\}\{4\}\{3\}\{1\}\{2\},$ $\{1\}\{2\}\{4\}\{3\}\{5\},\{4\}\{3\}\{1\}\{2\}\{5\})$. Then, the element with the highest chance of being in the first set of the partition is $5$, but the best response to agent $1$ selecting $3$ is for agent $2$ to select $4$, and similarly, the best response to agent $1$ selecting $2$ is for agent $2$ to select $1$.

\subsubsection{Only two processes, at least three agents}

\label{sec:twoprocesses}

We look now at the case where the parameters of the game are $(n, 2, Y_1(\lambda_1), Y_2(\lambda_2))$, and $n \geq 3$. Our main result in this section is the following:

\begin{theorem}
\label{th:twoprocesses}
Consider any Poisson-picking pool with parameters $(n, 2, Y_1(\lambda_1), Y_2(\lambda_2))$, and $n \geq 3$. Let $c$ be $\dfrac{\text{prob}(Y_1 > Y_2)}{\text{prob}(Y_1 < Y_2)}$. When $c \leq \dfrac{1}{n-1}$ or $ c\geq (n-1)$, it is the only symmetric Nash equilibrium to select the process with the highest rate. When $\dfrac{1}{n-1} < c < n-1$, all the symmetric Nash equilibria will have to satisfy the condition of being a root with odd multiplicity of the polynomial in $s(1)$ given by $\sum_{k=1} ^{n-2} {n -1 \choose k}  s(1)^{k} (1 - s(1))^{n-1-k} 
 \(  \frac{c*n}{k + 1}  -   \frac{n}{n -k}  \nonumber  \) + (1-s(1))^{n-1} (c*(n-1)  - 1 )  + (s(1))^{n-1} ( c - (n-1)   ).$

\end{theorem}

\begin{proof}

For each agent $i$, their strategy is uniquely encoded by the value of $s_i(1)$, which can move freely between $0$ and $1$. We have then that the probability that the choices of the agents are given by the string $w \in \{1,2\}^n$ is given by

\begin{align}
\label{probformula}
  \prod_{k=1}^n s_k(w_k)  =    \prod_{k=1}^n s_k(1)^{2-w_k} * (1 - s_k(1))^{w_k - 1}.
\end{align}

\noindent Also, denote by $n_x(w)$ the number of symbols equal to $x$ in a string $w$. Using this notation, as well as the formula in (\ref{probformula}), we have that for a fixed choice of strategies from each agent, the expected payoff $P_i$ for agent $i$ will be equal to

\begin{align*}
P_i & = \sum_{w \in \{1,2\}^n - \{1^n, 2^n\}}  \(  \prod_{k=1}^n s_k(1)^{2-w_k} * (1 - s_k(1))^{w_k - 1}  \) ( \text{prob}(Y_1 < Y_2)*\frac{n}{n_2(w)}*(w_i - 1) \\
& ~~  + \text{prob}(Y_1 = Y_2)*1 +  \text{prob}(Y_1 > Y_2)*\frac{n}{n_1(w)}*(2-w_i) -1 ).
\end{align*}

\noindent  This equation shows that whenever both processes are selected by at least someone, the pool will be distributed between those who select a process with the highest rate of the events. On the other hand, when everyone selects the same process, everyone just gets their money back. The derivative of $P_i$ with respect to the parameter $s_i(1)$  will be then  

\begin{align*}
\frac{d  P_i}{d s_i(1)} & = \frac{d}{d s_i(1)}  \sum_{w \in \{1,2\}^n - \{1^n,2^n\}}  \(  \prod_{\substack{k=1 \\ k \neq i } }^n s_k(1)^{2-w_k} * (1 - s_k(1))^{w_k - 1}  \)  * \\
& ~~ \( (2-w_i) s_i(1)^{1 - w_i} (1 - s_i(1))^{w_i - 1}  -   (w_i- 1) s_i(1)^{2-w_i} (1 - s_i(1))^{w_i - 2}  \) * \\
&  ~~  ( -1  +  \text{prob}(Y_1 < Y_2)*\frac{n}{n_2(w)}*(w_i - 1)  + \text{prob}(Y_1 = Y_2)*1 + \\
& ~~  \text{prob}(Y_1 > Y_2)*\frac{n}{n_1(w)}*(2-w_i) ).
\end{align*}

\noindent If we now only concern ourselves with symmetric equilibria, that means for any agent $j$ we can write the parameter $s_j(1)$ as $s(1)$. We can then write 

\begin{align*}
\frac{d  P_i}{d s_i(1)} & = \frac{d  P_i}{d s_i(1)} \sum_{w \in \{1,2\}^n - \{1^n,2^n\} }  \(  \prod_{\substack{k=1 \\ k \neq i } }^n s(1)^{2-w_k} * (1 - s(1))^{w_k - 1}  \)  * \\
& ~~ \( (2-w_i) s(1)^{1 - w_i} (1 - s(1))^{w_i - 1}  -   (w_i- 1) s(1)^{2-w_i} (1 - s(1))^{w_i - 2}  \) * \\
&  ~~  ( -1  +  \text{prob}(Y_1 < Y_2)*\frac{n}{n_2(w)}*(w_i - 1)  + \text{prob}(Y_1 = Y_2)*1 +\\
& ~~   \text{prob}(Y_1 > Y_2)*\frac{n}{n_1(w)}*(2-w_i) ) \\
\end{align*}

\noindent Now, if we consider both of the cases $w_i=1$ and $w_i=2$, we can verify that the contribution from taking derivatives in the second line will be equal to $1$ whenever $w_i=1$, and to $-1$ whenever $w_i=2$. If we take this into account, and split the outer sum into two different sums for the cases $w_i=1$ and $w_i=2$, we obtain the formula

\begin{align*}
\frac{d  P_i}{d s_i(1)} & = \sum_{w^{-i} \in \{1,2\}^{n-1} -\{1^{n-1},2^{n-1}\}}  s(1)^{n_1(w^{-i})} (1 - s(1))^{n_2(w^{-i})} * \\
&  ~~  \(   \text{prob}(Y_1 > Y_2)*\frac{n}{n_1(w^{-i}) + 1}  -   \text{prob}(Y_1 < Y_2)*\frac{n}{n_2(w^{-i}) + 1}  
\) \\
&  ~~+ (1-s(1))^{n-1} (-1 + \text{prob}(Y_1 > Y_2)*n +  \text{prob}(Y_1 = Y_2)*1 ) \\
&  ~~ - (s(1))^{n-1} (-1 + \text{prob}(Y_1 < Y_2)*n +  \text{prob}(Y_1 = Y_2)*1 ).
\end{align*}

\noindent Putting together all the terms corresponding to the same value of $n_1(w^{-i})= n - 1 - n_2(w^{-i}) $, and using that $ \text{prob}(Y_1 = Y_2) = 1 - \text{prob}(Y_1 < Y_2)  - \text{prob}(Y_1 > Y_2)$, we obtain that $\frac{d  P_i}{d s_i(1)} $  is given by

\begin{align} \label{eq:eqyatuib}
& \sum_{k=1} ^{n-2} {n -1 \choose k}  s(1)^{k} (1 - s(1))^{n-1-k} *  \nonumber \\ 
&   \(   \text{prob}(Y_1 > Y_2)*\frac{n}{k + 1}  -   \text{prob}(Y_1 < Y_2)*\frac{n}{n -k}  \nonumber  \)\\
&  + (1-s(1))^{n-1} ( \text{prob}(Y_1 > Y_2)*(n-1)  -  \text{prob}(Y_1 < Y_2) ) \nonumber \\
&   + (s(1))^{n-1} (  \text{prob}(Y_1 > Y_2)  - \text{prob}(Y_1 < Y_2)*(n-1)   ),
\end{align}

\noindent which is a polynomial of degree at most $n-1$ in $s(1)$.

Whenever $\text{prob}(Y_1 > Y_2)   \geq \text{prob}(Y_1 < Y_2)*(n-1)$, the results in Section \ref{sec:adjusting} mean that it will be a symmetric Nash equilibrium for everyone to select the first process. Indeed, if we examine \eqref{eq:eqyatuib}, we can see that whenever $\text{prob}(Y_1 > Y_2)  \geq \text{prob}(Y_1 < Y_2)*(n-1)$, and $s(1) < 1$, the value of \eqref{eq:eqyatuib} is greater than zero, since all the involved linear combinations of $ \text{prob}(Y_1 < Y_2)$ and $ \text{prob}(Y_1 > Y_2)$ are non-negative, and the one in the second-last line is certain to be positive. This means that for any symmetric strategy with $s(1) < 1$, it is the best response from agent $i$ to have $s_i(1)=1$. Therefore the only symmetric Nash equilibrium will be the one corresponding to $s(1)=1$. The same argument applies in the reverse case  $\text{prob}(Y_2 > Y_1)   \geq \text{prob}(Y_2 < Y_1)*(n-1)$.

It remains to consider the case where $ \text{prob}(Y_1 < Y_2)/(n-1) <  \text{prob}(Y_1 > Y_2) < \text{prob}(Y_1 < Y_2)*(n-1)$. Then, the analysis in Section \ref{sec:adjusting} shows that it will not be a symmetric Nash equilibrium for everyone to select the same process. We have then that the symmetric Nash equilibria (known to exist) will correspond to values of $s(1)$ in the interval $(0,1)$. Then, the fact that such an strategy has to be a local response to itself means that the partial derivative in \eqref{eq:eqyatuib}  has to be equal to zero if we plug in that value of $s(1)$. Taking into account the study of further derivatives of the payoff function, we have that not only does $s(1)$ need to be a root of  \eqref{eq:eqyatuib}, but also a root with odd multiplicity. Note that the polynomial  in   \eqref{eq:eqyatuib} is equivalent to the one in the statement of Theorem \ref{th:twoprocesses}, which is obtained through a division by $ \text{prob}(Y_2 < Y_1)$.

\end{proof}

We can look at the two small cases $n=3$ and $n=4$ to get a better understanding of what Theorem  \ref{th:twoprocesses} implies. When $n=3$, we obtain the following result:
 
\begin{corollary}
\label{col:3agents}
Consider any Poisson-picking pool with parameters $(3, 2, Y_1(\lambda_1), Y_2(\lambda_2))$. Let $c$ be $\dfrac{\text{prob}(Y_1 > Y_2)}{\text{prob}(Y_1 < Y_2)}$. When $c \leq \dfrac{1}{2}$ or $ c\geq 2$, it is the only symmetric Nash equilibrium for everyone to always select the process with the highest rate. When $\dfrac{1}{2} < c < 2$, the only symmetric Nash equilibrium will correspond to $s(1)= \dfrac{2c-1}{c+1}$.
\end{corollary}

\begin{proof}

\noindent  This corollary follows from considering Theorem \ref{th:twoprocesses}, and obtaining through mathematical software that for $n=3$ the value of \eqref{eq:eqyatuib} is given by

\begin{align*}
s(1)*(- \text{prob}(Y_1 > Y_2) - \text{prob}(Y_1 < Y_2)) + \text{prob}(Y_1 > Y_2)*2 - \text{prob}(Y_1 < Y_2).
\end{align*}

\noindent This is a linear equation in $s(1)$, whose only root is given by 

\begin{align*}
s(1) = \dfrac{ \text{prob}(Y_1 > Y_2)*2 - \text{prob}(Y_1 < Y_2)} {\text{prob}(Y_1 > Y_2) + \text{prob}(Y_1 < Y_2)} = \dfrac{2c-1}{c+1}.
\end{align*}

\end{proof}

\begin{corollary}

Consider any Poisson-picking pool with parameters $(4, 2, Y_1(\lambda_1), Y_2(\lambda_2))$. Let $c$ be $\dfrac{\text{prob}(Y_1 > Y_2)}{\text{prob}(Y_1 < Y_2)}$. When $c \leq \dfrac{1}{3}$ or $ c\geq 3$, it is the only symmetric Nash equilibrium for everyone to always select the process with the highest rate. When $\dfrac{1}{3} < c < 3$, the only symmetric Nash equilibrium will correspond to  $s(1) = \dfrac{1}{2}$ whenever $c=1$, and otherwise to  $s(1)~=~\dfrac{(3c + 1) -  \sqrt{ (3c + 1)^2 - 4(c-1)(3c-1)}}{2(c-1)}$.

\end{corollary}

\begin{proof}

We proceed similarly to the $n=3$ case, and obtain using mathematical software that the value of \eqref{eq:eqyatuib} is given by

\begin{align*}
& s(1)^2(\text{prob}(Y_1 >  Y_2) - \text{prob}(Y_1 < Y_2)) + s(1) * \(-\text{prob}(Y_1 > Y_2)*3 - \text{prob}(Y_1 < Y_2) \) \\  
+ &  ( \text{prob}(Y_1 > Y_2) *3- \text{prob}(Y_1 < Y_2)).
\end{align*}

\noindent Dividing by $\text{prob}(Y_1 < Y_2)$, we obtain a polynomial in $s(1)$ whose coefficients are a function of $c$. This polynomial is given by

\begin{align}
\label{eq:quadratic}
s(1)^2(c -1) + s(1) * \(-3c - 1 \)  + (3c - 1).
\end{align}

\noindent If we analyze the roots of this polynomial, we have that:

\begin{itemize}

\item If $c=1$, the only root is given by $s(1)=\dfrac{1}{2}$, as one might expect due to the symmetry of the situation.

\item If $\dfrac{1}{3} < c < 1$, the roots of the equation  are given by

\begin{align}
\frac{(3c + 1) \pm \sqrt{ (-3c -1)^2 - 4(c-1)(3c-1)}}{2(c-1)}.  \label{eq:thethingtwo}
\end{align}

The only one of these that is positive is

\[
\frac{\sqrt{ (3c + 1)^2 + 4(1-c)(3c-1)} - (3c+1)}{2(1 - c)},
\]

\noindent which will have to correspond then with the known to exist symmetric Nash equilibrium. 

\item If $1 < c < 3$, the roots of \eqref{eq:quadratic} are given again by \eqref{eq:thethingtwo}. However, this time, since $3c + 1 > 2c - 2$ and both values are positive, only one of the roots will be less than $1$, and it will again be the one given by

\[
\frac{(3c + 1) -  \sqrt{ (3c + 1)^2 - 4(c-1)(3c-1)}}{2(c-1)}
\]

\end{itemize}

\end{proof}

\begin{figure}[!htpb]
  \begin{center}
 	\includegraphics[height=7.25cm]{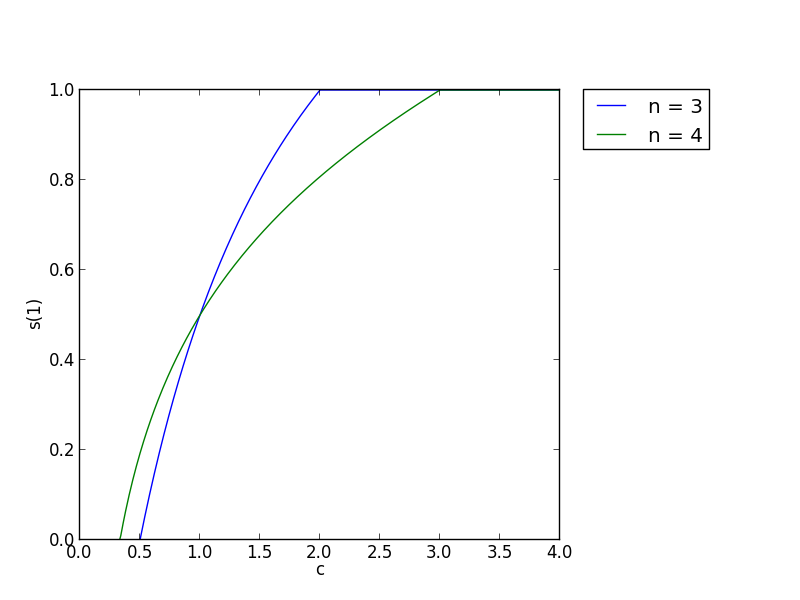}
  \end{center}
  \caption{Values for  $n=3$ and $n=4$ of $s(1)$ (the probability of choosing the first process in the symmetric Nash equilibrium) as a function of $c=\dfrac{\text{prob}(Y_1 > Y_2)}{\text{prob}(Y_1 < Y_2)}$  (the functions are equal to $1$ for for $c \geq n-1$). \label{fig:diagrams}}

\end{figure}

Motivated by what we observed in the cases with $3$ agents and $4$ agents, we make now the following conjecture:

\begin{conjecture}

Consider any Poisson-picking pool with parameters $(n, 2, Y_1(\lambda_1), Y_2(\lambda_2))$. Let $c$ be $\dfrac{\text{prob}(Y_1 > Y_2)}{\text{prob}(Y_1 < Y_2)}$. When $c \leq \dfrac{1}{n-1}$ or $ c\geq (n-1)$, it is the only symmetric Nash equilibrium for everyone to always select the process with the highest rate. When $\dfrac{1}{n-1} < c < (n-1)$, there will be only one symmetric Nash equilibrium,  with the corresponding value of $s(1)$ being a  continuous and monotonically increasing function of $c$.
\end{conjecture}

\noindent This conjecture represents the intuitive idea that as the process with the lower rate becomes more and more likely to be above the process with the higher rate, it makes more sense to put weight into it. It might be possible to prove the continuity part of it using the continuity of the roots of \eqref{eq:eqyatuib} as as function of $\text{prob}(Y_2 < Y_1)$ and $\text{prob}(Y_1 > Y_2)$. However, a naive argument in this direction runs into the issue that the roots of the corresponding polynomial can be complex. One could also use the fact that for every possible value of $s(1) \in (0, 1)$, there is exactly one value of $\dfrac{\text{prob}(Y_1 > Y_2)}{\text{prob}(Y_1 < Y_2)}$ that makes \eqref{eq:eqyatuib} equal to zero. This fact follows from fixing the value of $s(1)$ in \eqref{eq:eqyatuib}, and then looking at \eqref{eq:eqyatuib} as a linear expression in terms of $\text{prob}(Y_1 > Y_2)$ and $\text{prob}(Y_1 < Y_2)$.

Note that as in Section \ref{sec:adjusting}, the proof of Theorem \ref{th:twoprocesses} and its corollaries extends to the case of general winners-take-all pools, by looking at what the sign of $Y_2 - Y_1$ represents in terms of the partitions in Definition \ref{def:general}. One obtains then the following theorem:

\begin{theorem}
Consider any winners-take-all pool with parameters $(n, 2, X)$, and $n \geq 3$. Let $c$ be $\dfrac{\text{prob}(X =\{1\}\{2\})}{\text{prob}(X = \{2\}\{1\})}$. When $c \leq \dfrac{1}{n-1}$ or $ c\geq (n-1)$, it is the only symmetric Nash equilibrium to select the option with a higher chance of being in the first set of the partition. When $\dfrac{1}{n-1} < c < n-1$, all the symmetric Nash equilibria will have to satisfy the condition of being a root with odd multiplicity of the polynomial in $s(1)$ given by $\sum_{k=1} ^{n-2} {n -1 \choose k}  s(1)^{k} (1 - s(1))^{n-1-k} 
 \(  \frac{c*n}{k + 1}  -   \frac{n}{n -k}  \nonumber  \) + (1-s(1))^{n-1} (c*(n-1)  - 1 )  + (s(1))^{n-1} ( c - (n-1)   ).$
\end{theorem}

\noindent Following this pattern, it is straightforward to generalize our corollaries and conjecture as well.

\section{Numerical study}

\label{sec:practicalAnalysis}
 
We use the Neng Python package \cite{Sebek2013} to conduct the numerical study in this section. This software gives us a mixed Nash equilibrium for a game given given its expected payoff matrix. To obtain these expected payoff matrices, we use Monte Carlo simulation. Note that the finding of mixed Nash equilibria does not generally have any known fast algorithms (and there is in fact evidence for the absence of such an algorithm \cite{daskalakis2009complexity}), which limits the practical usage of general solvers to cases with a small number of agents and processes.

The software we use cannot be constrained to return a symmetric equilibrium and it is not deterministic, in the sense that it can return different mixed Nash equilibria in different runs when applied to the same payoff matrix.  Our choice of diversification metric is then derived from associating with each Poisson process a sample average probability that it is chosen. To do so, we run the equilibrium-finding software $t$ times, and in each run we compute for each process the probability that an agent chosen uniformly at random will pick it in the equilibrium distribution that we obtained. Then, we average the $t$ resulting distributions \footnote{The plots here correspond to the case where $t=100$. From observing the convergence process, the absolute error in the resulting probabilities seems to be between $0.005$ and $0.02$.}. While this process is not fully rigorous, it does pass the sanity check of putting half of the weight on each process when we have two processes, at least $3$ agents, and each of the processes has the same rate. Moreover, it gives rise to trends that generally seem sensible.

We will start from the case where there are $3$ agents, $3$ processes, and their rates are $1$, $k$, and $k^2$ ($k< 1$), and see how the behavior of the assigned weights/probabilities evolves as we tweak the parameters. First, we consider the case with $k=0.95$, and observe what happens if we increase the number of agents or the number of processes. In the latter case, this is done by adding another process with rate again equal to $k$ times the rate of the last process. We obtain then the following plot:

\includegraphics[width=12cm]{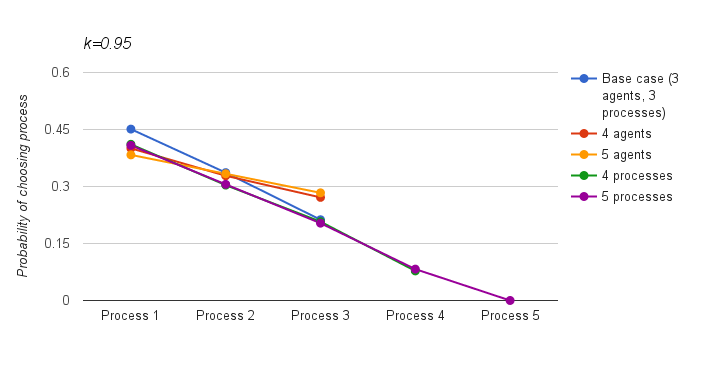}

One can observe that when we add a fourth agent, there is a loss of weight  (i.e. probability) for process $1$, and a corresponding increase in the weight for process $3$. However, adding a fifth agent seems to only barely increase the weight removed from process $1$, and given our margin of error of around $0.02$, it is not even clear that it does result in such a change. We believe that this is due to the fact that in the change to $4$ agents, our probabilities already get a lot closer to being uniform. One can also see (the green line is behind the purple line) that adding a fourth process with a rate of $(0.95)^3$ introduces more diversification in the answers of the agents, with the weight on process $4$ coming from processes $1$ and $2$. However, there is no more diversification when we introduce a fifth process with a smaller rate of $(0.95)^4$.

If we now repeat this process with $k=0.875$, we obtain the following graph:

\includegraphics[width=12cm]{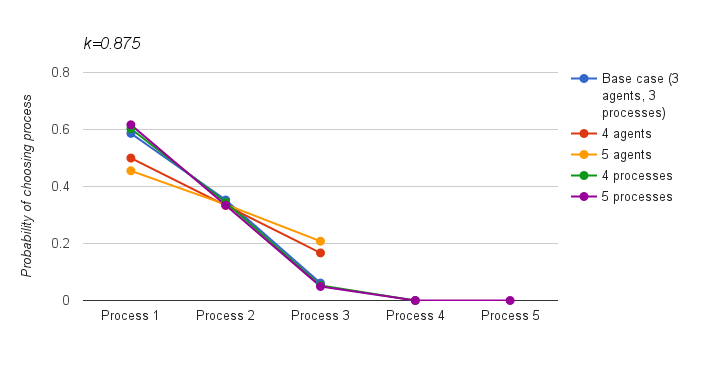}

With respect to the addition of agents, we see a similar pattern as before, except that in this situation there is less of a diminishing returns behavior in the introduction of a fifth agent. Note also that we start further from the uniform distribution to begin with. With respect to the addition of processes, the rates of the new processes $k^3=(0.875)^3$ and $k^4=(0.875)^4$ are not enough to increase the pressure for diversification.

If we further decrease the value of $k$, taking it down to $k=0.80$, we obtain:

\includegraphics[width=12cm]{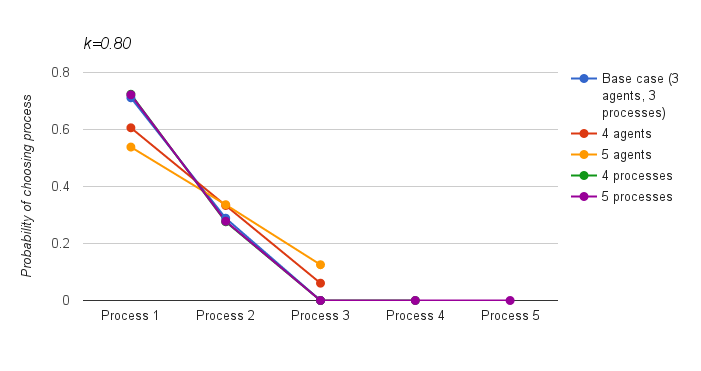}

We can see that initially there is no weight assigned to process $3$  (the blue and green lines are mostly hidden behind the purple line). However, after we add a fourth agent, there is a decrease in the probability assigned to process $1$, and corresponding increases  (more or less equal to each other)  to the probabilities for process $2$ and process $3$. Adding a fifth agent induces a smaller decrease of the weight assigned to process $1$, which this time seems to go entirely to process $3$.

The next value of $k$, equal to $0.725$, results in:

\includegraphics[width=12cm]{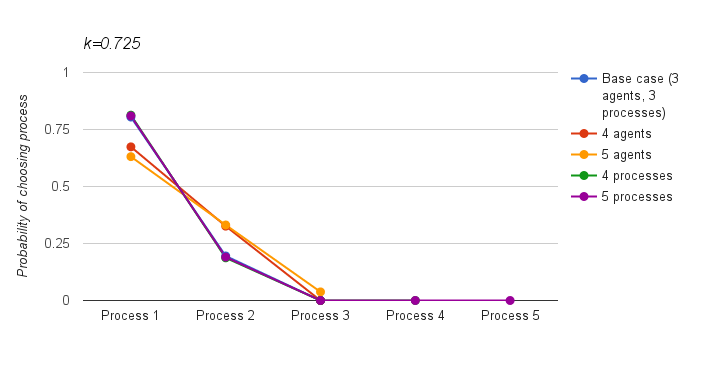}

Now there is again no weight initially assigned to process $3$. However, this time adding a fourth agent does not change this fact, and instead it significantly increases the probability assigned to process $2$. It is only when a fifth agent is introduces that the probability for process $3$ becomes non-zero, with the corresponding weight coming from process $1$.  

For the last value of $k$, equal to $0.650$, we obtain:

\includegraphics[width=12cm]{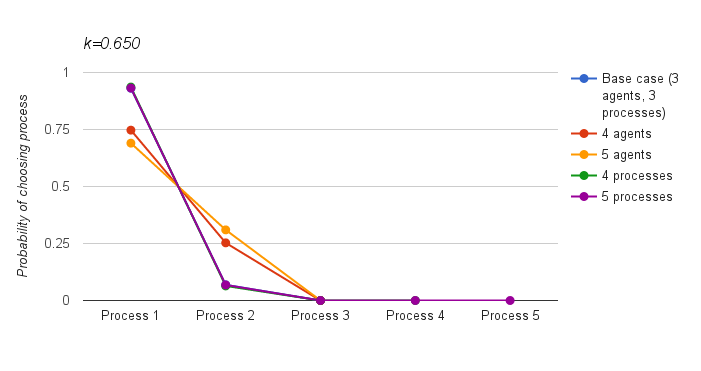}

This time, even adding a fifth agent is not enough for process $3$ to be given a non-zero weight. Instead, adding a fourth agent results in changes similar to the previous plot, and adding a fifth process results in an additional small increase for the probability for process $2$, at the expense of the weight for process $1$.

Overall, we can see in these plots distinct regimes with respect to the change in the probability assigned to a process that occurs after the introduction of more agents: once the probability is close enough to a final equilibrium value, its change when more agents are added is non-existing or very small. On the other hand, when the probability assigned to the process is still at $0$, it might take the addition of several agents before the probability starts moving. In the middle of these two regions, the probabilities seem to move quite fast.

This generalizes what we saw in Section \ref{sec:twoprocesses} when looking at the situation with just two processes. For example, we can observe in Figure \ref{fig:diagrams} that there is a region where adding a third agent does not make it optimal in terms of symmetric Nash equilibria to deviate from selecting the favorite, but adding a fourth agent does finally induce some pressure for diversification. However, it remains unclear which values does this diminishing return process converge to as the number of agents $n$ goes to infinity.  Our intuition is that in the case of just two processes, the probability of choosing a process will converge to the probability of that process having more events than the other, relative to the probability of no ties between them. Then, when more processes are added, this formula might need to be modified to take into account the probabilities associated with each of the partial tie situations.

The behavior we have seen is also consistent with a successive unloading of weight from the process(es) with the highest rate to those with a smaller weight. That is to say, it is consistent with the behavior where pressure to diversification first induces a tendency to pick the second best process instead of the first one, until a certain bound is fulfilled, and then it begins to mostly induce a tendency to pick the third best process instead of the first one, and so on.

We have also seen that  it was relevant to introduce a fourth process only  in the first case, where the rate for that fourth process was closest to the rates for the other three processes. We conjecture that the behavior with respect to the introduction of new processes is different to the behavior with respect to the introduction of new agents. What we mean by this is that we saw a behavior in our graphs that is consistent with the idea that after introducing a large enough number of agents, in the Nash equilibria there will be a non-zero probability assigned to each process. However, for processes we believe that if adding an extra process with rate $r$ does not lead to any diversification of choices in the Nash equilibria, the addition of any number of processes with rates $\leq r$ will also not lead to any diversification

As far as the rates of the processes are involved, one would also expect that given some absolute differences of rates between the processes, there will  be more diversification going on in an optimal choice between these processes when the relative difference of rates is smaller. This is because intuitively, two processes with rates given by 99.5 and 100 are closer to each other than two processes with rates 0.5 and 1.0. One can also appeal to the results in Section \ref{sec:twoprocesses} to justify our intuition, since the corresponding value of $c$ (the relative odds of the first process doing better versus it doing worse) will be given by a value close to $1$ in the former case, and by a value around $2.5$ in the latter case. To further verify our intuition, we offset  the rates in our initial experiment ($k=0.95$) by additive factors of $1.0$ and $-0.5$, and observe the results:

\includegraphics[width=12cm]{firstGraph.png}

\includegraphics[width=12cm]{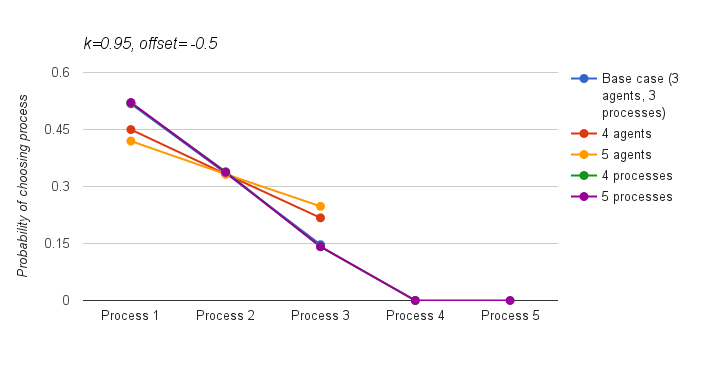}

\includegraphics[width=12cm]{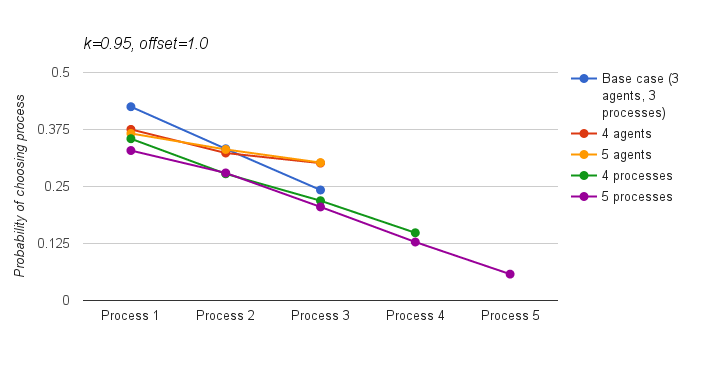}

As we can see, it does indeed hold that in the case  with larger relative differences between rates ($\text{offset}=-0.5$) there is less diversification that in the original $k=0.95$ graph. Similarly, there is more diversification in the case  with smaller relative differences between rates. We can also see that this phenomenon affects all five lines in the graph.

\section{Discussion}

We have looked from several angles at an example of the pressure for diversification in winners-take-all situations.  A general theme is that as the number of agents increases, it becomes more and more valuable to be the unique winner, since the size of the pool increases. Therefore, it makes more sense to engage in risk-seeking behavior to become this unique winner. However, we have also seen that there is also a property of diminishing returns with respect to the total number of agents as far as this behavior is concerned. One can see these properties for example in our very first result in Section \ref{sec:adjusting}, where the threshold that determines whether one should go with the underdog option decreases proportionally to $\frac{1}{n}$. One can also see it in the fact that any symmetric strategy is a Nash equilibrium for $n=2$, or in the plots that we obtain in Section \ref{sec:practicalAnalysis}, which show us a regime where  the number of agents is not enough to entice differentiation, a regime where a larger number of agents will not entice further differentiation, and an intermediate regime.

A careful reader will realize that outside of our graphs, we did not explicitly use in Section \ref{sec:anstudy} the fact that the independent variables $Y_1, \ldots, Y_m$ for the number of events correspond to Poisson distributions. One could use this fact in a theoretical analysis, and express quantities like $\text{prob}(Y_2 > Y_1)$, with $Y_1$ and $Y_2$ independent Poisson processes, as functions of the corresponding rates $\lambda_1$ and $\lambda_2$. This would lead to analytical results similar to the ones we obtain in Section \ref{sec:anstudy}, but expressed explicitly in terms of the rates of the Poisson processes under consideration.

Regarding further work, in Section \ref{sec:adjusting} it might be of interest to characterize the optimal responses of an agent to a wider variety of greedy choices from the other agents. Similarly, it would also be useful to obtain a closed formula for the roots of the polynomial in Theorem \ref{th:twoprocesses}.

It might be of interest as well to consider the difference in terms of diversification between the not-necessarily-symmetric Nash equilibria that we analyze in Section \ref{sec:practicalAnalysis}, and the restriction to symmetric Nash equilibria. More formally, one could consider the probability associated with each process by sampling mixed Nash equilibria uniformly at random, and then players at random. Then, one could repeat this process considering only symmetric Nash equilibria, and compare the results. Our numerical exploration seems to possibly have found a case where these two measures would have values that are close to each other, but different. This is the case where there are $2$ processes, with rates $1.25$ and $1$, and $3$ agents. The only symmetric Nash equilibrium corresponds to selecting the first process with probability $\approx 0.76$, as can be obtained from Corollary \ref{col:3agents}. However, our numerical exploration found other mixed Nash equilibria, and suggests that the average probability with which the first process is selected in a Nash equilibrium is equal to $\approx 0.70 < 0.76$.

It is also possible to slightly tweak our Poisson-picking setting, and study Poisson-picking pools with non-independent Poisson processes. As we mentioned in the introduction, this would allow us to better model the situation where a group of friends are betting about who will be the top goal-scorer in a competition, rather than on the top goal-scoring team. Additionally, one could also extend our analysis to the case where each agent chooses $k$ out of $m$ Poisson processes, and the counts from the chosen processes are added together when determining the payoffs. This would result in a model closer to the context for daily fantasy sports, studied for example in \cite{hunterpicking}. 

Another real-life example that our analysis could extend to is that of football pools betting, with global sales on the order of magnitude of a billion of dollars as recently as 2007 \cite{forrest2013football}. In these games, each contestant is trying to guess outcomes for a selection of association football games, and those who are the most accurate get to split the pool between the them. This is similar to the winners-take-all situation we studied in this game, and leads again to a pressure for diversification.

Along these lines of adapting our model more closely to real-life situations, one could also try to arrive to behavioral models of real-life participants in parimutuel games. This would likely involve modeling the participants using several distinct sub-models, according to their goals (which might not necessarily be to maximize their expected income) and their level of sophistication. This approach is taken by the work in \cite{bayraktar2017high}, which answers in the negative the question of whether a particular simple stratified model results in some of the phenomena that can be observed in real-life parimutuel games.

\section*{Acknowledgement}
Thanks for useful feedback are due to Kate Larson, for whose CS886 offering this research was originally conducted.

\bibliographystyle{alpha}
\bibliography{arxivReferences}

\newcommand{\etalchar}[1]{$^{#1}$}
\begin{thebibliography}{{\"O}WCC11}

\bibitem[BM17]{bayraktar2017high}
Erhan Bayraktar and Alexander Munk.
\newblock High-roller impact: A large generalized game model of parimutuel
  wagering.
\newblock {\em Market Microstructure and Liquidity}, 3(01):1750006, 2017.

\bibitem[CHN14]{chawla2014mechanism}
Shuchi Chawla, Jason Hartline, and Denis Nekipelov.
\newblock Mechanism design for data science.
\newblock In {\em Proceedings of the 15th ACM conference on Economics and
  computation}, pages 711--712. ACM, 2014.

\bibitem[DC97]{dixon1997modelling}
Mark~J Dixon and Stuart~G Coles.
\newblock Modelling association football scores and inefficiencies in the
  football betting market.
\newblock {\em Journal of the Royal Statistical Society: Series C (Applied
  Statistics)}, 46(2):265--280, 1997.

\bibitem[DFP{\etalchar{+}}06]{dulleck2006all}
Uwe Dulleck, Paul Frijters, Konrad Podczeck, et~al.
\newblock {\em All-pay auctions with budget constraints and fair insurance}.
\newblock Inst. f{\"u}r Volkswirtschaftslehre, Johannes Kepler Univ. Linz,
  2006.

\bibitem[DGP09]{daskalakis2009complexity}
Constantinos Daskalakis, Paul~W Goldberg, and Christos~H Papadimitriou.
\newblock The complexity of computing a nash equilibrium.
\newblock {\em SIAM Journal on Computing}, 39(1):195--259, 2009.

\bibitem[FP13]{forrest2013football}
David Forrest and Levi P{\'e}rez.
\newblock The football pools.
\newblock In {\em The Oxford Handbook on the Economics of Gambling}. 2013.

\bibitem[HM95]{hurley1995note}
William Hurley and Lawrence McDonough.
\newblock A note on the hayek hypothesis and the favorite-longshot bias in
  parimutuel betting.
\newblock {\em The American Economic Review}, pages 949--955, 1995.

\bibitem[HVZ16]{hunterpicking}
David~Scott Hunter, Juan~Pablo Vielma, and Tauhid Zaman.
\newblock Picking winners using integer programming.
\newblock {\em arXiv preprint 1604.01455}, 2016.

\bibitem[Mah82]{maher1982modelling}
Michael~J Maher.
\newblock Modelling association football scores.
\newblock {\em Statistica Neerlandica}, 36(3):109--118, 1982.

\bibitem[Mul77]{mullet1977simeon}
Gary~M Mullet.
\newblock Simeon poisson and the national hockey league.
\newblock {\em The American Statistician}, 31(1):8--12, 1977.

\bibitem[Nas51]{nash1951non}
John Nash.
\newblock Non-cooperative games.
\newblock {\em Annals of mathematics}, pages 286--295, 1951.

\bibitem[{\"O}WCC11]{ostling2011testing}
Robert {\"O}stling, Joseph Tao-yi Wang, Eileen~Y Chou, and Colin~F Camerer.
\newblock Testing game theory in the field: Swedish lupi lottery games.
\newblock {\em American Economic Journal: Microeconomics}, 3(3):1--33, 2011.

\bibitem[Pen04]{pennock2004dynamic}
David~M Pennock.
\newblock A dynamic pari-mutuel market for hedging, wagering, and information
  aggregation.
\newblock In {\em Proceedings of the 5th ACM conference on Electronic
  commerce}, pages 170--179. ACM, 2004.

\bibitem[PWY03]{plott2003parimutuel}
Charles~R Plott, Jorgen Wit, and Winston~C Yang.
\newblock Parimutuel betting markets as information aggregation devices:
  experimental results.
\newblock {\em Economic Theory}, 22(2):311--351, 2003.

\bibitem[Rob85]{robbins1985some}
Herbert Robbins.
\newblock Some aspects of the sequential design of experiments.
\newblock In {\em Herbert Robbins Selected Papers}, pages 169--177. Springer,
  1985.

\bibitem[Seb13]{Sebek2013}
Petr Sebek.
\newblock Neng: A tool for computing nash equilibria.
\newblock \url{https://github.com/Artimi/neng}, 2013.

\bibitem[Tho07]{thomas2007inter}
Andrew~C Thomas.
\newblock Inter-arrival times of goals in ice hockey.
\newblock {\em Journal of Quantitative Analysis in Sports}, 3(3), 2007.

\bibitem[TZ88]{thaler1988anomalies}
Richard~H Thaler and William~T Ziemba.
\newblock Anomalies: Parimutuel betting markets: Racetracks and lotteries.
\newblock {\em The Journal of Economic Perspectives}, 2(2):161--174, 1988.

\end{thebibliography}

\end{document}